\documentclass[journal]{IEEEtran}
\usepackage{amsmath,amsfonts}
\usepackage{algorithmic}
\usepackage{algorithm}
\usepackage{array}
\usepackage{textcomp}
\usepackage{stfloats}
\usepackage{url}
\usepackage{verbatim}
\usepackage{cite}
\usepackage{setspace}
\usepackage{amsthm}
\usepackage{amssymb,amsmath,latexsym}
\usepackage{color,cite,times}
\usepackage{epstopdf}
\usepackage{multirow}
\usepackage{array}
\usepackage{booktabs}
\usepackage{graphicx}
\usepackage{subcaption}
\newtheorem{assumption}{Assumption}

\newtheorem{theorem}{\textbf{Theorem}}

\allowdisplaybreaks[4]

\DeclareMathOperator*{\diag}{diag}

\ifCLASSINFOpdf
\else
\fi
\begin{document}
%
\title{
 Spectral-Convergent Decentralized Machine Learning: Theory and Application in Space Networks
}    
%
%
%

\author{Zhiyuan Zhai, 
Shuyan Hu, \IEEEmembership{Member, IEEE},
Wei Ni, \IEEEmembership{Fellow, IEEE},\\
Xiaojun Yuan, \IEEEmembership{Senior Member, IEEE}, and 
        Xin Wang, \IEEEmembership{Fellow, IEEE}
	}

\maketitle

\begin{abstract}
Decentralized machine learning (DML) supports  collaborative training in large-scale networks  with no  central server. It is sensitive to the quality and reliability of inter-device communications that result in time-varying and stochastic topologies. This paper studies the impact of unreliable communication on the convergence of DML and establishes a direct connection between the spectral properties of the mixing process and the global performance. We provide rigorous convergence guarantees under random topologies and derive bounds that characterize the impact of the expected mixing matrix's spectral properties on learning. We formulate a spectral optimization problem that minimizes the spectral radius of the expected second-order mixing matrix to enhance the convergence rate under probabilistic link failures.  To solve this non-smooth spectral problem in a fully decentralized manner, we design an efficient subgradient-based algorithm that integrates Chebyshev-accelerated eigenvector estimation with local update and aggregation weight adjustment, while ensuring symmetry and stochasticity constraints without central coordination. Experiments on a realistic low Earth orbit (LEO) satellite constellation with time-varying inter-satellite link models and real-world remote sensing data demonstrate the feasibility and effectiveness of the proposed method. The method significantly improves classification accuracy and convergence efficiency compared to existing baselines, validating its applicability in satellite and other decentralized systems.
\end{abstract}


%
\IEEEpeerreviewmaketitle

\section{Introduction}

Decentralized machine learning (DML) has emerged as a promising paradigm for distributed model training in large-scale, heterogeneous, and infrastructure-less networks~\cite{liu2022decentralized}. In DML, each device locally updates a model using its private dataset and communicates with its peers for collaborative aggregation, thereby avoiding raw data transmission and supporting privacy-preserving learning over flexible topologies\cite{li2020federated}. Compared with  classical centralized federated learning (FL)  requiring periodic coordination with a central server, DML eliminates the single-point-of-failures  and offers improved scalability and resilience. This decentralized architecture is particularly beneficial in scenarios where central coordination is unavailable or undesirable, such as wireless sensor networks~\cite{akyildiz2002wireless}, mobile edge systems~\cite{mach2017mobile}, ad hoc communication networks~\cite{ren2007information}, and space-based platforms~\cite{10818523}.

\subsection{Motivation and Challenges}
DML has found applications across a wide range of distributed systems. Representative scenarios include wireless sensor networks deployed in remote areas, vehicular ad hoc networks for collaborative perception, and mobile edge computing systems with highly dynamic user mobility \cite{yuan2024decentralized}. 
One representative use case of DML involves communication-constrained environments such as satellite constellations, where centralized training is often impractical. For instance, low Earth orbit (LEO) satellites generate large volumes of sensing data but face limited ground connectivity. Onboard DML can alleviate downlink bottlenecks and enable timely model updates, especially with recent advances in space-grade AI hardware~\cite{zhu2017deep, li2020edge,Asheralieva2025Dynamic}.

Despite its advantages, DML faces significant technical challenges. Particularly, the lack of a central server introduces difficulties in achieving global consensus, especially when the communication topology is time-varying or partially connected~\cite{koloskova2020unified}. Dynamic link conditions, bandwidth constraints, and stochastic communication failures can severely hinder synchronization and convergence of models in DML~\cite{ye2022decentralized}. Moreover, under such unstable conditions, decentralized optimization often suffers from slow convergence, as information mixing becomes inefficient and local models drift apart. This makes fast convergence a critical requirement, especially in time-sensitive or resource-constrained environments. Furthermore, the efficiency of model aggregation under such uncertain environments depends on the network structure and aggregation weights, which are difficult to optimize without central coordination~\cite{liu2022deep, kungurtsev2023decentralized}. These challenges highlight the need for robust and topology-aware DML algorithms that can operate efficiently and converge rapidly in unreliable communication networks.

\subsection{Related Work}

Existing research on DML has yielded theoretical and practical advancements regarding algorithm design, system optimization, and application.
For example, convergence analyses were provided in~\cite{lian2017can, che2022decentralized}
under non independent and identically distributed (non-IID) data distributions and partial client participation.
Communication-efficient methods were developed to reduce bandwidth requirements while maintaining model accuracy, e.g., gradient quantization~\cite{reisizadeh2020fedpaq},
sparse aggregation \cite{tang2022gossipfl},
and adaptive client selection mechanisms~\cite{nishio2019client}.

Some existing works have focused on  aggregation mechanisms and topology optimization for DML
to enhance  efficiency, scalability, and convergence.
Lian et al. \cite{lian2017can} pioneered the Decentralized Parallel Stochastic Gradient Descent (D-PSGD) algorithm, replacing the central server with point-to-point model aggregation.
By formulating the problem as a consensus optimization with doubly stochastic mixing matrices, they achieved $\mathcal{O}(1/\sqrt{T})$ convergence for non-convex objectives while eliminating single point of failures, where $T$ is the number of training rounds. 
Li et al. \cite{li2021topology} proposed a spectral graph-theoretic approach for topology optimization, proving that the spectral gap of the communication graph directly impacts convergence speed. Their greedy algorithm dynamically rewired connections to maximize the algebraic connectivity, and accelerate  convergence in ring topologies.
Khan et al. \cite{khan2024graph} formulated topology optimization as a constrained graph learning problem. Their differentiable graph neural network (GNN) optimizer jointly learns node embeddings and edge weights to maximize convergence speed under bandwidth constraints, demonstrating faster convergence and lower communication costs in large-scale IoT deployments.
Li et al.~\cite{li2024adaptive} proposed an adaptive DML framework tailored for  device networks.
Dynamically adjusting intra-plane and inter-plane aggregation strategies and introducing a self-compensation mechanism to mitigate unreliable cross plane communication, their approach achieves robust and communication-efficient convergence under dynamic orbital topologies.

However, these existing studies have often overlooked the design of DML  under unreliable communication conditions, where links may intermittently fail or exhibit highly variations. 
Such conditions are prevalent in real-world systems like satellite constellations or ad hoc networks, where communication is frequently intermittent, bandwidth-constrained, and subject to environmental disruptions or dynamic topology changes. 
Many existing approaches, e.g., \cite{many1,many2,many3,roy2019braintorrent,Li2025Biasing,Lin2025FedSN}, have relied on partial centralization for tasks, such as topology control, synchronization, or global aggregation, limiting their applicability in infrastructure-less deployments. 
These limitations hinder the scalability and robustness of DML in real-world decentralized environments.
While our earlier work~\cite{zhai_distributed} proposed a distributed design for decentralized machine learning by optimizing aggregation weights across devices, it relied on power iteration to estimate the dominant eigenvector and used constrained convex optimization to restore feasibility at each round. However, this design faces two key limitations in practice. First, power iteration converges slowly when the spectral gap is small, which is common in sparse or weakly connected communication topologies, thereby increasing the number of local iterations and prolonging convergence. Second, the projection step requires solving a quadratic program with inequality constraints, introducing considerable computational and coordination overhead, especially in large-scale or latency-sensitive networks.


\subsection{Contribution}
This paper proposes a fully decentralized DML  framework  to operate over time-varying  and unreliable  communication topologies. 
Specifically, we formulate decentralized stochastic gradient descent (SGD)  using a matrix-based representation that explicitly captures random link failures and dynamic point-to-point connectivity. 
We  express the decentralized updates of DML as a coupled recursion over model and topology evolution, and establish convergence guarantees for non-convex objectives.
Our analysis reveals that the convergence rate of DML depends on the second-largest eigenvalue modulus of the expected mixing matrix, which reflects the network's mixing efficiency under stochastic topologies. 
To accelerate learning in such environments, we design a fully distributed subgradient  algorithm that adaptively adjusts the mixing weights to minimize the expected spectral radius, requiring only local interactions without global knowledge. 
This algorithm enables network-aware model aggregation in unreliable, large-scale, infrastructure-less systems.

The key contributions of this paper are summarized as follows:
\begin{itemize}
    \item \textit{DML under stochastic communication:} We propose a fully decentralized DML  framework that operates over time-varying  and unreliable communication topologies, without  central coordination. The framework accounts for random link failures and dynamic neighbor changes, enabling scalable learning in infrastructure-less and intermittently connected networks.


    \item \textit{Distributed spectral optimization of aggregation weights:} In light of our analysis, we formulate a spectral radius minimization problem over the expected mixing matrix and solve the problem using a fully decentralized subgradient  algorithm. Our approach leverages matrix differential theory, distributed eigenvector estimation, and feasibility restoration, allowing each node to adapt its aggregation weights using only local information.
    
     \item \textit{Accelerated spectral optimization for fast convergence:} 
    We accelerate the spectral optimization by replacing power iteration with a Chebyshev polynomial-based method, which converges faster and is insensitive to small spectral gaps. To reduce complexity, we further replace constrained projections with a lightweight normalization step. These improvements jointly enhance convergence speed and scalability in dynamic, resource-limited networks.

    \item \textit{Extensive evaluation over real-world LEO dataset:} We validate our framework using a Starlink-like LEO constellation and the EuroSAT remote sensing dataset. Experiments show that our algorithm enables fully onboard training across devices with intermittent inter-device links, significantly improving image classification accuracy, while approaching centralized performance without any ground station involvement.
\end{itemize}

The rest of this paper is organized as follows. Section II introduces the system model, including the DML protocols and the probabilistic modeling of unreliable communication links.
Section III provides the convergence analysis of DML under non-convex objectives and establishes theoretical guarantees under stochastic topologies.
Section IV formulates the spectral optimization problem to enhance convergence, and develops a fully decentralized subgradient algorithm incorporating distributed eigenvector estimation and feasibility restoration.
Section V presents experimental results  using the EuroSAT dataset, validating the effectiveness of our method under various settings.
Section VI concludes the paper.

\textbf{Notation:}  
Italic letters denote scalar variables. Bold lowercase and uppercase letters represent vectors and matrices, respectively. $(\cdot)^\top$ denotes the matrix transpose. $\diag(\cdot)$ constructs a diagonal matrix from a vector, and $\mathrm{Diag}(\mathbf{A})$ preserves the diagonal entries of a square matrix $\mathbf{A}$ while zeroing out all off-diagonal elements. $\odot$ denotes the Hadamard (element-wise) product. $\|\cdot\|$ denotes the Euclidean norm, $\|\cdot\|_2$ denotes the spectral norm, and $\|\cdot\|_F$ denotes the Frobenius norm. $|\cdot|$ denotes either the absolute value of a scalar or the cardinality of a set, depending on context. $\mathbf{1}$ denotes the all-one vector, and $\mathbf{I}$ denotes the identity matrix. $\rho(\cdot)$ denotes the spectral radius of a matrix, and $\lambda_k(\cdot)$ denotes the $k$-th largest eigenvalue. $\mathbb{E}[\cdot]$ denotes expectation.

\section{System Model}

\subsection{DML over Unreliable Links}

Consider a DML framework implemented across $N$ devices. Each device is equipped with  sensing, storage, and computation units, and  collects local data from its respective footprint. The objective of DML is to collaboratively train a shared machine learning model to minimize the average loss function:
\begin{equation}
    \mathcal{L}(\boldsymbol{w}) = \frac{1}{N} \sum_{i=1}^{N} \mathcal{L}_i(\boldsymbol{w}),
\end{equation}
where $\boldsymbol{w} \in \mathbb{R}^d$ represents the global model parameters, and the local loss function at device $i$ is defined as
\begin{equation}
    \mathcal{L}_i(\boldsymbol{w}) := \mathbb{E}_{\zeta_i \sim \mathcal{S}_i} \ell(\boldsymbol{w}, \zeta_i),
\end{equation}
with $\mathcal{S}_i$ being the local dataset, and $\ell(\cdot, \cdot)$ denoting the loss over sample $\zeta_i$.

At each training round $t$, the decentralized protocol functions as follows:
\begin{itemize}
    \item \textbf{Local update}: Each device $i$ samples a mini-batch $\zeta_i^{(t)} \in \mathcal{S}_i$ and computes the stochastic gradient $\nabla \ell(\boldsymbol{w}_i^{(t)}, \zeta_i^{(t)})$ at its current local model $\boldsymbol{w}_i^{(t)}$.

    \item \textbf{Model dissemination}: Devices exchange models via point-to-point communication links. Let $m_{ji}^{(t)}$ be the indicator function of successful transmission. If $m_{ji}^{(t)} = 1$, the model from device $j$ is received successfully by device $i$; otherwise, device $i$ reuses its own model $\boldsymbol{w}_i^{(t)}$. The received model vector is
    \begin{equation}
        \tilde{\boldsymbol{w}}_{ji}^{(t)} = m_{ji}^{(t)} \boldsymbol{w}_j^{(t)} + (1 - m_{ji}^{(t)}) \boldsymbol{w}_i^{(t)}.
    \end{equation}

    \item \textbf{Model fusion}: Device $i$ aggregates models from neighbors as follows:
    \begin{equation}
        \boldsymbol{w}_i^{(t + \frac{1}{2})} = \boldsymbol{w}_i^{(t)} + \sum_{j =1,j\neq i}^{N} a_{ij} m_{ji}^{(t)} (\boldsymbol{w}_j^{(t)} - \boldsymbol{w}_i^{(t)}),
    \end{equation}
    where $a_{ij}$ is the aggregation coefficient assigned by device \( j \) to the model of device \( i\).

    \item \textbf{Local model update}: Each device updates its model by applying SGD, as given by
    \begin{equation}
        \boldsymbol{w}_i^{(t+1)} = \boldsymbol{w}_i^{(t + \frac{1}{2})} - \eta \cdot \nabla \ell(\boldsymbol{w}_i^{(t)}, \zeta_i^{(t)}),
    \end{equation}
    where $\eta$ is the learning rate.
\end{itemize}
This protocal enables fully distributed learning without central coordination. 

\subsection{Probabilistic Modeling of Communication Links}
\label{subsec:link_model}

Consider intermittent link transmissions caused by misalignment, hardware inaccuracies, and environmental factors such as signal attenuation. 
Let $\mathbf{B}_0$ denote the ideal adjacency matrix, where $\mathbf{B}_0[i,j] = 1, \forall i \neq j$, representing full inter-device connectivity under ideal conditions without any physical or environmental impairments, and $\mathbf{B}_0[i,i] = 0, \forall i$. The actual adjacency matrix at round $t$ can be modeled as a randomly perturbed version of $\mathbf{B}_0$, as given by
\begin{equation}
    \mathbf{B}^{(t)} = \mathbf{B}_0 \odot \mathbf{M}^{(t)},
\end{equation}
where $\odot$ denotes the Hadamard (element-wise) product, and $\mathbf{M}^{(t)} \in \{0,1\}^{N \times N}$ is a  binary matrix capturing the availability of each link at round $t$ with $m_{ij}^{(t)}$ being the $(i,j)$-th entry.

Let $q_{ij} \in [0,1]$ denote the probability of successful transmission from node $i$ to node $j$, which is modeled as a Bernoulli sampling process of link $(i,j)$:
\begin{equation}
    q_{ij}=\Pr\left(m_{ij}^{(t)} = 1\right), \quad \forall (i,j).
\end{equation}
Note that  $m_{ii}^{(t)} = 0,\forall i$, $m_{ij}^{(t)} = m_{ji}^{(t)}$, and $q_{ij}=q_{ji}, \forall i,j$ indicating symmetric point-to-point communication link conditions.
We assume that the link status indicators are independent among different device pairs\footnote{Symmetric point-to-point communication link conditions are common in bidirectional communication settings, while statistical independence is reasonable when device separations or channel fading are uncorrelated~\cite{andrews2005laser}.}; i.e., $m_{ij}^{(t)}$ and $m_{kl}^{(t)}$ are statistically independent for any $(i,j) \neq (k,l)$.

\subsection{Global Matrix-Form View of Decentralized Updates}

To facilitate a systematic analysis of the DML process across the entire network, 
we propose to consider the collective evolution of model parameters across all participating devices. Let $\boldsymbol{W}^{(t)} \triangleq \left[\boldsymbol{w}_1^{(t)}, \dots, \boldsymbol{w}_N^{(t)}\right] \in \mathbb{R}^{d \times N}$ denote the matrix formed by stacking the local models at round $t$, and define the corresponding stochastic gradient matrix as
\begin{equation}
\nabla \ell^{(t)} \triangleq \left[ \nabla \ell(\boldsymbol{w}_1^{(t)}, \zeta_1^{(t)}), \dots, \nabla \ell(\boldsymbol{w}_N^{(t)}, \zeta_N^{(t)}) \right] \in \mathbb{R}^{d \times N}.
\end{equation}

Then, the decentralized update process across the network can be expressed in the following compact matrix form:
\begin{equation}
\boldsymbol{W}^{(t+1)} = \boldsymbol{W}^{(t)} \mathbf{P}^{(t)} - \eta \nabla \ell^{(t)},
\label{eq:matrix_update}
\end{equation}
where $\mathbf{P}^{(t)} \in \mathbb{R}^{N \times N}$ is a time-varying mixing matrix encoding the effect of neighbor aggregation coefficients and stochastic link availability at round $t$, as given by
\begin{equation}
\mathbf{P}^{(t)} = \mathbf{I} + \mathbf{A} \odot \mathbf{M}^{(t)} - \mathrm{Diag}(\mathbf{A} \mathbf{M}^{(t)}),
\label{eq:mixing_matrix}
\end{equation}
where $\mathbf{A}$ is the weight matrix with the  $(i,j)$-th entry  $a_{ij}$.

The randomness in this  process stems from the stochasticity of inter-device communication and the random sampling of training data. The conditional expectation of \eqref{eq:matrix_update} given the current models $\boldsymbol{W}^{(t)}$ and sampled data $\boldsymbol{\zeta}^{(t)}$ yields
\begin{equation}
\mathbb{E} \left[ \boldsymbol{W}^{(t+1)} \mid \boldsymbol{W}^{(t)}, \boldsymbol{\zeta}^{(t)} \right] = \boldsymbol{W}^{(t)} \overline{\mathbf{P}} - \eta \nabla \ell^{(t)},
\end{equation}
where $\overline{\mathbf{P}} = \mathbb{E}\{\mathbf{P}^{(t)}\}$ is the expected mixing matrix, with its entries given by
\begin{equation}
\overline{p}_{ij} =
\begin{cases}
a_{ij} q_{ij}, &\text{if} ~i \neq j; \\
1 - \sum_{j \neq i} a_{ij} q_{ij}, &\text{if} ~i = j.
\end{cases}
\label{eq:expected_mixing}
\end{equation}

\section{Spectral Analysis of Convergence}\label{sec3}

To facilitate the convergence analysis of the DML process under stochastic inter-device connectivity, the following assumptions are considered.

\begin{assumption}[Lipschitz Gradient Regularity]\label{as:lip}
Each local objective $\mathcal{L}_i(\boldsymbol{w})$ is continuously differentiable, and its gradient is Lipschitz continuous with constant $L>0$, i.e.,
\begin{equation}
\|\nabla \mathcal{L}_i(\boldsymbol{w}) - \nabla \mathcal{L}_i(\boldsymbol{v})\| \leq L \|\boldsymbol{w} - \boldsymbol{v}\|, \quad \forall \boldsymbol{w}, \boldsymbol{v} \in \mathbb{R}^d.
\end{equation}
\end{assumption}

\begin{assumption}[Bounded Gradient Discrepancy]\label{as:grad_var}
The variance introduced by data sampling and model heterogeneity is uniformly bounded. That is, there exist constants $\sigma^2>0$ and $ \delta^2 > 0$ such that $\forall i \in [N],  \boldsymbol{w} \in \mathbb{R}^d$,
\begin{align}
\mathbb{E} \left[ \|\nabla \ell(\boldsymbol{w}, \zeta_i) - \nabla \mathcal{L}_i(\boldsymbol{w}) \|^2 \right] &\leq \sigma^2; \\
\mathbb{E} \left[ \|\nabla \mathcal{L}_i(\boldsymbol{w}) - \nabla \mathcal{L}(\boldsymbol{w}) \|^2 \right] &\leq \delta^2.
\end{align}
\end{assumption}

\begin{assumption}[Consensus Mixing Condition]\label{as:mixing_matrix}
The aggregation coefficient matrix $\mathbf{A}$ used in model fusion is symmetric and doubly stochastic, i.e., $\mathbf{A}^\top = \mathbf{A}$ and $\mathbf{A}\boldsymbol{1} = \boldsymbol{1}$. Let $\overline{\mathbf{P}^2} \triangleq \mathbb{E}\left[ (\mathbf{P}^{(t)})^2 \right]$ denote the second-order moment of the mixing matrix $\mathbf{P}^{(t)}$. The spectral norm of the non-leading eigenmodes is strictly smaller than one, i.e.,
\[
\rho(\overline{\mathbf{P}^2})=\max \left\{ |\lambda_2(\overline{\mathbf{P}^2})|, \dots, |\lambda_N(\overline{\mathbf{P}^2})| \right\} < 1.
\]
\end{assumption}

\smallskip

Assumptions~\ref{as:lip}--\ref{as:mixing_matrix} are standard in  the convergence analysis of SGD under decentralized settings~\cite{lian2017can,koloskova2019decentralized}. Assumption~1 ensures the smoothness of each local objective function, which guarantees that the gradient does not change abruptly.
In Assumption~\ref{as:grad_var}, the constants $\sigma^2$ and $\delta^2$ quantify the stochastic gradient noise and the level of statistical heterogeneity of devices, respectively.
For Assumption~\ref{as:mixing_matrix}, it is known that a doubly stochastic matrix \( \mathbf{W} \) has the largest eigenvalue \( \lambda_1(\mathbf{W}) = 1 \), and all eigenvalues satisfy \( |\lambda_i(\mathbf{W})| \leq 1 \). Assumption~3 tightens this by requiring that \( |\lambda_i(\mathbf{W})| < 1 \) for all \( i \neq 1 \).
This requirement  guarantees the geometric decay of disagreement among devices \cite{lian2017can}. It ensures that the non-leading modes of the mixing process contract over time, allowing the local models to asymptotically agree and thereby enabling convergence of the global objective.

Under these assumptions,  the following new theorem about the convergence  of  DML  is established. 

\begin{theorem}[Ergodic Convergence under Stochastic Links]\label{thm:convergence}
Under Assumptions~\ref{as:lip}–\ref{as:mixing_matrix}, if the learning rate satisfies
\[
\eta < \frac{1 - \sqrt{ \rho(\overline{\mathbf{P}^2}) }}{6 L \sqrt{N}},
\]
 the gradient norm of the loss function  at the global average model $\bar{\boldsymbol{w}}^{(t)}$ admits the following convergence bound:
\begin{align}
&\frac{1}{T} \sum_{t=0}^{T-1} \mathbb{E} \left\| \nabla \mathcal{L} \left( \bar{\boldsymbol{w}}^{(t)} \right) \right\|^2
\leq \frac{1}{\left(\frac{1}{2} - 9 \Gamma(\overline{\mathbf{P}^2}) \right)} \notag \\
&\quad \times \left( \frac{\mathcal{L}_0 - \mathcal{L}^*}{\eta T} + \frac{\eta L \sigma^2}{2N} + \sigma^2 \Gamma(\overline{\mathbf{P}^2}) + 9 \delta^2 \Gamma(\overline{\mathbf{P}^2}) \right),
\label{eq:convergence_bound}
\end{align}
where $\bar{\boldsymbol{w}}^{(t)} \triangleq \frac{1}{N} \sum{i=1}^N \boldsymbol{w}_i^{(t)}$ is the global average model at round $t$, $T$ is the number of training rounds, $\mathcal{L}_0 = \mathcal{L}(\bar{\boldsymbol{w}}^{(0)})$ is the initial global loss, $\mathcal{L}^*$ is the optimal value of the global objective, and $\Gamma(\overline{\mathbf{P}^2}) = \frac{N \eta^2 L^2}{(1 - \sqrt{ \rho(\overline{\mathbf{P}^2}) })^2 - 18 N \eta^2 L^2}$.
The expectation on the LHS of \eqref{eq:convergence_bound} is taken over  the random  communication link realizations and  data sampling.
\end{theorem}
\begin{proof}
See Appendix~A.
\end{proof}

\section{Decentralized Aggregation Optimization}
\subsection{Problem Formulation}

Theorem~\ref{thm:convergence} indicates that the asymptotic convergence rate of the DML process is determined by a spectral property of the expected second-order mixing matrix $\overline{\mathbf{P}^2}$. Specifically, the upper bound on the gradient norm of the global average model depends monotonically on $\rho(\overline{\mathbf{P}^2})$, with a larger $\rho(\overline{\mathbf{P}^2})$ leading to slower consensus among devices and, consequently, slower overall convergence.
There is an opportunity to optimize the aggregation coefficient matrix $\mathbf{A}$ to accelerate the convergence of DML, as $\mathbf{A}$ directly affects the construction of $\mathbf{P}^{(t)}$ as defined in~\eqref{eq:mixing_matrix}. 
To this end, we formulate a spectral optimization problem to enhance DML performance:
\begin{align}
\min_{\mathbf{A}} \quad & \rho(\overline{\mathbf{P}^2}) \label{eq:problem_formulation} \\
\text{s.t.} \quad & \mathbf{A} = \mathbf{A}^\top, \quad \mathbf{A} \mathbf{1} = \mathbf{1}, \quad a_{ij} \geq 0,\ \forall i,j, \notag
\end{align}
where the constraints ensure that $\mathbf{A}$ remains symmetric and doubly stochastic, as indicated in Assumption~\ref{as:mixing_matrix}.

Two critical challenges arise from~\eqref{eq:problem_formulation}, including

\subsubsection{Absence of Central Coordination}

The intermittent connectivity and the 
 inherent decentralization render conventional aggregation coefficient optimization methods~\cite{9154332,10506083,9716792,9916128,9783194} unsuitable, as they typically require centralized access to full network state information.
A distributed optimization framework is needed, in which each device adjusts its local aggregation coefficients based solely on its local or neighbors' information. Such a framework  is expected to collectively steer DML toward  improved convergence, even in the absence of central control or global synchronization.

\subsubsection{Complexity of  Spectral Objective}

The  objective in~\eqref{eq:problem_formulation}, namely minimizing $\rho(\overline{\mathbf{P}^2})$, involves the spectral radius of the expected second-order mixing matrix. This matrix, $\overline{\mathbf{P}^2} = \mathbb{E}[(\mathbf{P}^{(t)})^2]$, captures the nontrivial effect of time-varying link availability and aggregation weights, as $\mathbf{P}^{(t)}$ depends nonlinearly on both the stochastic adjacency matrix $\mathbf{M}^{(t)}$ and the weight matrix $\mathbf{A}$. 
Direct optimization of $\rho(\overline{\mathbf{P}^2})$ is analytically intractable, as evaluating a spectral radius  requires centralized eigenvalue computations and global matrix statistics~\cite{boyd2004convex}.

To overcome these challenges, we approximate $\rho(\overline{\mathbf{P}^2})$ with a tractable surrogate objective that retains sensitivity to network connectivity and mixing quality, while supporting distributed minimization. In the next subsection, we derive such a surrogate objective and design a decentralized subgradient algorithm that enables each device to iteratively refine its local aggregation coefficients using only local information and statistical link transmission patterns.

\subsection{Surrogate Reformulation}

We design a tractable surrogate objective based on spectral norm bounds that preserves the key convergence characteristics of problem \eqref{eq:problem_formulation}.
The spectral radius of $\overline{\mathbf{P}^2}$ excluding the consensus eigenvalue can be  expressed as
\begin{align}
    \rho(\overline{\mathbf{P}^2}) = \left\| \overline{\mathbf{P}^2} - \frac{\mathbf{1}\mathbf{1}^\top}{N} \right\|_2,
\end{align}
where $\|\cdot\|_2$ denotes  spectral norm \cite{yoshida2017spectral}. 

Define the deviation between the second-order moment $\overline{\mathbf{P}^2}$ and the squared mean $\overline{\mathbf{P}}^2$ as
\begin{align}
    \boldsymbol{\Delta} = \overline{\mathbf{P}^2} - \overline{\mathbf{P}}^2.
\end{align}
Using triangle inequality, we obtain the upper bound of the nontrivial spectral radius $\rho(\overline{\mathbf{P}^2})$, as given by
\begin{align} \label{eq:spectral_norm_bound}
    \rho(\overline{\mathbf{P}^2}) \leq \left\| \overline{\mathbf{P}}^2 - \frac{\mathbf{1}\mathbf{1}^\top}{N} \right\|_2 + \|\boldsymbol{\Delta}\|_2.
\end{align}
The first term on the right-hand side (RHS) of \eqref{eq:spectral_norm_bound}  captures the spectral contraction behavior due to the expected mixing topology. The second term on the RHS of \eqref{eq:spectral_norm_bound} quantifies the perturbation caused by random fluctuations.

In large-scale systems with many devices, the contribution of each individual model can be negligible. No single device significantly influences the aggregation.
In this case, the entries of $\mathbf{P}^{(t)}$ vary independently over time, and their variances scale as $\mathcal{O}(1/N^2)$ \cite{zhu2022topology}; the Frobenius norm of $\boldsymbol{\Delta}$ satisfies
\begin{align}
    \|\boldsymbol{\Delta}\|_F^2 \!=\!\! \sum_{i,j=1}^N \left( \mathbb{E}[P^{(t)}_{ij} P^{(t)}_{ij}] \!-\! (\mathbb{E}[P^{(t)}_{ij}])^2 \right)^2 \!\!\!= \!\mathcal{O}(\frac{1}{N}).
\end{align}
By standard inequalities of norms, the spectral norm satisfies $\|\boldsymbol{\Delta}\|_2 \leq \|\boldsymbol{\Delta}\|_F = \mathcal{O}(1/\sqrt{N})$~\cite{horn1990matrix}. In the large-scale regime with $N \to \infty$, the perturbation term vanishes, i.e.,
\begin{align}
    \lim_{N \to \infty} \|\boldsymbol{\Delta}\|_2 \to 0.
\end{align}

Combining this with the bound in \eqref{eq:spectral_norm_bound}, we arrive at
\begin{align}
    \rho(\overline{\mathbf{P}^2}) \lessapprox \left\| \overline{\mathbf{P}}^2 - \frac{\mathbf{1}\mathbf{1}^\top}{N} \right\|_2 = \left(\rho(\overline{\mathbf{P}})\right)^2,
\end{align}
where $\rho(\overline{\mathbf{P}}) = \max\{\lambda_2(\overline{\mathbf{P}}), -\lambda_N(\overline{\mathbf{P}})\}$ gives the nontrivial spectral radius of the expected mixing matrix. 
In this sense, minimizing $\rho(\overline{\mathbf{P}})$ serves as an effective surrogate for minimizing the original objective $\rho(\overline{\mathbf{P}^2})$ in \eqref{eq:problem_formulation}.

This surrogate objective not only reduces analytical complexity but also enables decentralized algorithm design. Since $\overline{\mathbf{P}}$ depends only on first-order link statistics, e.g., mean availability (or reliability), it can be estimated and controlled using local observations. We henceforth adopt $\rho(\overline{\mathbf{P}})$ as the tractable optimization objective in the remainder of this paper.

\subsection{Subgradient Analysis}
\label{SA}

With the surrogate objective $\overline{\mathbf{P}}$, we  now convert \eqref{eq:problem_formulation} to the following  optimization problem:
\begin{equation}\label{pro111}
\begin{aligned}
    \min_{\mathbf{A}} \quad & \rho(\overline{\mathbf{P}}) \\
    \text{s.t.} \quad & \mathbf{A}^\top = \mathbf{A}, \quad \mathbf{A}\mathbf{1} = \mathbf{1}, \quad \mathbf{A} \in [0,1]^{N \times N},
\end{aligned}
\end{equation}
where the goal is  to minimize the nontrivial spectral radius $\rho(\overline{\mathbf{P}})$ of the expected mixing matrix $\overline{\mathbf{P}}$ by optimizing the aggregation matrix $\mathbf{A}$.

The spectral radius $\rho(\overline{\mathbf{P}})$ can be characterized via the following variational formulations:
\begin{align}
    \lambda_2(\overline{\mathbf{P}}) &= \sup_{\substack{\mathbf{u}^\top \mathbf{1} = 0 \\ \|\mathbf{u}\|_2 \leq 1}} \mathbf{u}^\top \overline{\mathbf{P}} \mathbf{u},\label{26} \\
    -\lambda_N(\overline{\mathbf{P}}) &= \sup_{\|\mathbf{u}\|_2 \leq 1} -\mathbf{u}^\top \overline{\mathbf{P}} \mathbf{u}.\label{27}
\end{align}
Since both \eqref{26} and \eqref{27} are pointwise suprema of linear functions, their maximum is a convex function of $\mathbf{A}$.

For ease of exposition, we define a linear surrogate operator:
\begin{equation}
    \mathbf{R}(\mathbf{A}) := \mathbf{I} + \frac{1}{2} \sum_{i,j=1}^N a_{ij} \mathbf{E}_{ij},
\end{equation}
where each structured matrix $\mathbf{E}_{ij}$ is defined element-wise as
\begin{equation}
    [\mathbf{E}_{ij}]_{kl} =
    \begin{cases}
        q_{ij}, & \text{if} ~(k,l) \in \{(i,j), (j,i)\}; \\
        -q_{ij}, & \text{if} ~(k,l) \in \{(i,i), (j,j)\}; \\
        0, & \text{otherwise}.
    \end{cases}
\end{equation}
Hence, we can recast problem \eqref{pro111} as
\begin{equation}\label{pro1112}
\begin{aligned}
    \min_{\mathbf{A}} \quad & \rho(\mathbf{R}(\mathbf{A})) \\
    \text{s.t.} \quad & \mathbf{A}^\top = \mathbf{A}, \quad \mathbf{A}\mathbf{1} = \mathbf{1}, \quad \mathbf{A} \in [0,1]^{N \times N},
\end{aligned}
\end{equation}

Let $\lambda_*$ denote the dominant nontrivial eigenvalue of $\mathbf{R}(\mathbf{A})$, and  $\mathbf{v}$ be the associated unit eigenvector. Then, the directional derivative of $\rho(\mathbf{R}(\mathbf{A}))$ with respect to $a_{ij}$ is given by
\begin{equation}
    \nabla_{a_{ij}} \rho(\mathbf{R}(\mathbf{A})) = \frac{1}{2} \cdot \mathbf{v}^\top \mathbf{E}_{ij} \mathbf{v}.
\end{equation}
This leads to two cases for the subgradient:
\begin{itemize}
    \item If $\rho(\mathbf{R}(\mathbf{A})) = \lambda_2(\mathbf{R}(\mathbf{A}))$, then
    \begin{equation}
        \nabla_{a_{ij}} \rho(\mathbf{R}(\mathbf{A})) = -\frac{1}{2} q_{ij} (v_i - v_j)^2;
    \end{equation}
    \item If $\rho(\mathbf{R}(\mathbf{A})) = -\lambda_N(\mathbf{R}(\mathbf{A}))$, then
    \begin{equation}
        \nabla_{a_{ij}} \rho(\mathbf{R}(\mathbf{A})) = \frac{1}{2} q_{ij} (v_i - v_j)^2.
    \end{equation}
\end{itemize}
Therefore, the complete subgradient matrix $\nabla \rho(\mathbf{R}(\mathbf{A}))$ is 
\begin{equation}\label{aaa12}
    \nabla_{a_{ij}} \rho(\mathbf{R}(\mathbf{A})) =
    \begin{cases}
        -\frac{1}{2} q_{ij} (v_i - v_j)^2, & \text{if } \lambda_2 \text{ is active}, \\
        \frac{1}{2} q_{ij} (v_i - v_j)^2, & \text{if } \lambda_N \text{ is active}.
    \end{cases}
\end{equation}

This structure reveals that the subgradient with respect to each aggregation weight $a_{ij}$ depends solely on the link reliability $q_{ij}$ and the difference between the corresponding eigenvector entries $v_i$ and $v_j$.
Each device $i$ can compute its local subgradient $\nabla_{a_{ij}} \rho(\mathbf{R}(\mathbf{A}))$ in \eqref{aaa12} using only its local and neighbors' information. 
This property enables a fully decentralized subgradient descent method that requires only local measurement/observation or information.

\subsection{Distributed Eigenvector Estimation}
\label{dec}

To enable decentralized subgradient computation in Section~\ref{SA}, each device must estimate the eigenvector associated with the second-largest eigenvalue of the expected mixing matrix $\overline{\mathbf{P}}$; see \eqref{aaa12}. While conventional power iteration methods have been widely used for this purpose~\cite{zhai_distributed}, they typically suffer from slow convergence due to the spectral proximity between $\lambda_2$ and the trivial eigenvalue $\lambda_1 = 1$ \cite{montijano2012chebyshev}. To overcome this limitation, we adopt a Chebyshev-accelerated iterative approach~\cite{winkelmann2019chase}, which provides faster spectral separation and is compatible with decentralized implementation.

Given that $\overline{\mathbf{P}}$ is symmetric and doubly stochastic, we define a residual matrix that removes the dominant eigencomponent:
\begin{align}
	\widetilde{\mathbf{P}} = \overline{\mathbf{P}} - \frac{\mathbf{1}\mathbf{1}^\top}{N},
\end{align}
so that the leading eigenvalue of $\widetilde{\mathbf{P}}$ corresponds to the second-largest eigenvalue of $\overline{\mathbf{P}}$.

Chebyshev iteration  works efficiently when the eigenvalues of the matrix lie in the interval $[-1, 1]$. However, the eigenvalues of $\widetilde{\mathbf{P}}$ may lie outside this range. To address this, we rescale $\widetilde{\mathbf{P}}$ so that all  eigenvalues fall within $[-1, 1]$, which improves the numerical stability and convergence speed of the Chebyshev recurrence. Given estimated bound $\mu > \nu$ on the nonzero eigenvalues of $\widetilde{\mathbf{P}}$, we define the rescaled matrix:
\begin{align}
	\mathbf{T} = \frac{2\widetilde{\mathbf{P}} - (\mu + \nu)\mathbf{I}}{\mu - \nu},
\end{align}
which maps the spectrum of $\widetilde{\mathbf{P}}$ into the  interval $[-1, 1]$.

Based on this, the Chebyshev recurrence for estimating the target eigenvector $\mathbf{v}$ (as required in~\eqref{aaa12}) is summarized in Algorithm~\ref{alg:chebyshev}.
Notably, all matrix-vector products with $\mathbf{T}$ in Algorithm~\ref{alg:chebyshev} can be computed using only local communications, since $\widetilde{\mathbf{P}}$ inherits the sparsity  of the underlying network topology. The normalization step (Line 5) can be performed through standard distributed averaging methods, such as consensus protocols~\cite{xiao2020survey}. Additionally, the recurrence requires only two prior states to be stored locally at each device, resulting in low memory overhead; see Line 4.

Upon convergence, the final iterate $\mathbf{v}^{(K)}$ provides a reliable estimate of the nontrivial eigenvector of $\overline{\mathbf{P}}$ and can be directly used for decentralized subgradient evaluation in the aggregation weight optimization process.

\begin{algorithm}[t]
    \caption{Distributed Chebyshev Iteration for Eigenvector Estimation}
    \label{alg:chebyshev}
    \begin{algorithmic}[1]
        \STATE \textbf{Input:} Estimate bounds $\mu > \nu$, number of iterations $K$.
        \STATE \textbf{Initialize:} Each device $i$ randomly selects $v_i^{(0)}$, computes $v_i^{(1)} = \mathbf{T} v_i^{(0)}$.
        \FOR{$k = 2$ to $K$}
            \STATE Each device updates its local state:
            \[
            \mathbf{v}^{(k)} = 2 \mathbf{T} \mathbf{v}^{(k-1)} - \mathbf{v}^{(k-2)}.
            \]
            \STATE Normalize $\mathbf{v}^{(k)}$ via distributed averaging.
        \ENDFOR
        \ENSURE $\mathbf{v}^{(K)}$ as the estimated eigenvector.
    \end{algorithmic}
\end{algorithm}

\subsection{Symmetric Doubly Stochastic Guarantee}
\label{dpmc}

Following each subgradient update, the aggregation matrix $\mathbf{A}$ may deviate from the feasible set $\mathcal{S}$, defined as
\[
\mathcal{S} = \left\{ \mathbf{A} \in \mathbb{R}^{N \times N} ~\middle|~ \mathbf{A}^\top = \mathbf{A},\ \mathbf{A} \mathbf{1} = \mathbf{1},\ \mathbf{A} \geq 0 \right\},
\]
which specifies that $\mathbf{A}$ must remain symmetric, row-stochastic, and nonnegative.

To ensure feasibility throughout the optimization, we develop a decentralized adjustment mechanism that restores $\mathbf{A}$ to a valid structure after each update. This mechanism consists of two fully local operations: (i) symmetry enforcement, and (ii) row normalization.

\textit{Step 1: Symmetry Enforcement.}  
Each device first symmetrizes its local aggregation weights through pairwise exchanges with neighbors. Given an intermediate matrix $\mathbf{A}$, symmetry is enforced via
\[
\mathbf{A}_{\text{sym}} = \frac{1}{2} \left( \mathbf{A} + \mathbf{A}^\top \right),
\]
which ensures that $\mathbf{A}_{\text{sym}}$ is symmetric. This step preserves the network sparsity pattern and requires only bidirectional communication between neighboring devices.

\textit{Step 2: Local Row Normalization.}  
Each device rescales its local weights to satisfy the row-stochastic condition. For each neighbor $ j $, device $i$ performs
\[
a_{ij}^{\text{new}} = \frac{\mathbf{A}_{\text{sym}}(i,j)}{\sum_{k=1}^{N} \mathbf{A}_{\text{sym}}(i,k)},
\]
where $\mathcal{N}_i$ is the neighborhood of device $i$. The resulting matrix $\mathbf{A}_{\text{new}} = [a_{ij}^{\text{new}}]$ satisfies all feasibility requirements.

This two-step adjustment ensures that the matrix $\mathbf{A}$ retains the desired symmetric and doubly stochastic structure throughout the decentralized optimization process.

\subsection{Overall Decentralized Subgradient Algorithm}

We now integrate the key modules developed in the preceding sections to construct a fully decentralized subgradient method for solving the spectral minimization problem in~\eqref{pro111}. The resulting procedure—outlined in Algorithm~\ref{alg:leoDML}—combines the subgradient computation designed in Section~\ref{SA}, the eigenvector estimation  described in Section~\ref{dec}, and the feasibility restoration  introduced in Section~\ref{dpmc}.

\begin{algorithm}[t]
	\caption{Decentralized Subgradient Algorithm for Spectral Mixing Optimization}
	\label{alg:leoDML}
	\begin{algorithmic}[1]
		\STATE \textbf{Input:} Initial mixing matrix $\mathbf{A}(0)$ with non-negative entries, link reliability profile $\{q_{ij}\}$, step size $\gamma > 0$, total iterations $J_{\max}$.
		\STATE \textbf{Initialization:} Set $n = 0$.
		\FOR{$n = 0$ to $J_{\max} - 1$}
			\STATE Estimate the nontrivial eigenvector $\mathbf{v}(n)$ associated with $\rho(\mathbf{R}(\mathbf{A}))$ via the Chebyshev-accelerated procedure; see Section~\ref{dec}.
			\STATE Each device $i$ computes its local subgradient entries $\{g(a_{ij})\},{\forall j }$ using \eqref{aaa12}, based on $v_i$, $v_j$, and link statistics $\{q_{ij}\}$.
			\STATE Update the  weights via subgradient descent:
			\[
			a_{ij}^{\text{temp}}(n+1) = a_{ij}(n) - \gamma \cdot g(a_{ij}(n)).
			\]
			\STATE Apply the decentralized feasibility adjustment; see Section~\ref{dpmc}:
			\[
			\mathbf{A}(n+1) \xleftarrow{}(\mathbf{A}^{\text{temp}}(n+1)).
			\]
		\ENDFOR
		\ENSURE Optimized aggregation matrix $\mathbf{A}^\star = \mathbf{A}(J_{\max})$.
	\end{algorithmic}
\end{algorithm}

Algorithm \ref{alg:leoDML} is fully decentralized and requires only local message exchange and computation,  aligning well with practical device networks operating under dynamic and infrastructure-less conditions. The final matrix $\mathbf{A}^\star$ minimizes the spectral mixing radius, thereby accelerating the convergence of DML.
The time complexity of Algorithm~\ref{alg:leoDML} is $\mathcal{O}(J_{\max} K N)$, where $J_{\max}$ is the number of subgradient iterations, and $K$ denotes the number of Chebyshev iterations for eigenvector estimation. This linear scaling with respect to the network size makes the algorithm computationally feasible for large-scale decentralized systems.

\section{Numerical Results}
\subsection{Simulation Setup}

To assess the performance of the proposed decentralized subgradient algorithm for DML under realistic channel conditions, we simulate an orbital plane within an LEO satellite constellation. Following the Starlink Phase I configuration~\cite{mcdowell2020low}, the orbital plane consists of $N = 22$ satellites, randomly distributed along a circular orbit.
The link success probability $q_{ij}$ between two devices $i$ and $j$ is determined by the following three physical factors: (i) inter-satellite distance, (ii) beam pointing deviation, and (iii) environmental interference.

These effects are captured by the following model:
\begin{equation}
q_{ij} = 1 - \max \left\{
    \frac{\alpha_d d_{ij}}{d_{\max}},\;
    \frac{\alpha_\theta \theta_{ij}}{\theta_{\max}},\;
    w_{ij}
\right\},
\label{eq:link-successful transmission-model}
\end{equation}
where $d_{ij}$ denotes the Euclidean distance between devices $i$ and $j$, and $d_{\max}$ is the maximum permissible link distance. We set  $d_{\max}=3,000$ km. $d_{ij}$ is computed via great-circle arcs. $\theta_{ij}$ is the beam steering angle deviation from device $j$ to device $i$, while $\theta_{\max} = 60^\circ$ is the maximum allowable angular offset for reliable alignment.  $\theta_{ij}$ is derived from the relative angular orientation of each node’s orbital tangent vectors. The coefficient $w_{ij} \in [0,1]$ accounts for environment-dependent interference effects, e.g., atmospheric disturbance or orbital position degradation~\cite{grubsky2020effects}. The positive constants $\alpha_d$ and $\alpha_\theta$ normalize the relative influence of distance and angular deviation, respectively~\cite{wilson2005free}.
Unless otherwise specified, we use the following default parameter values:
$
\alpha_d = 0.7, \, \alpha_\theta = 0.8$, and $ w_{ij} = 0.05$.

We conduct DML experiments on a simulated satellite constellation using the EuroSAT remote sensing dataset~\cite{helber2019eurosat}. The dataset contains 27,000 geo-referenced RGB images captured by the Sentinel-2 satellite, each annotated with one of 10 land cover or land use classes.  An overview of the dataset is illustrated in Fig.~\ref{fig:eurosat_samples}.
We adopt a lightweight convolutional neural network (CNN) architecture, termed {LightCNN\_EuroSAT}, to perform multi-class classification. The network consists of three convolutional blocks, followed by two fully connected layers, with a total parameter count under one million.
The global dataset is partitioned into $N = 22$ local subsets, one per satellite. Each satellite is assigned an IID local dataset with the sample size  drawn uniformly from $\mathrm{Unif}(100, 125)$, reflecting statistical variations across orbital regions. The participating satellites collaboratively train a global model using the decentralized learning algorithm described in Section~II. 

\begin{figure}[htbp]
\centering

\begin{subfigure}[t]{0.19\linewidth}
    \includegraphics[width=\linewidth]{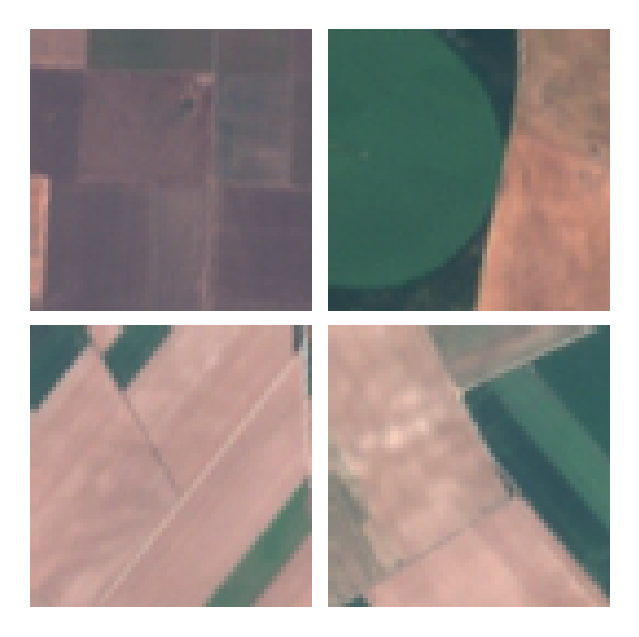}
    \caption*{Annual Crop}
\end{subfigure}
\begin{subfigure}[t]{0.19\linewidth}
    \includegraphics[width=\linewidth]{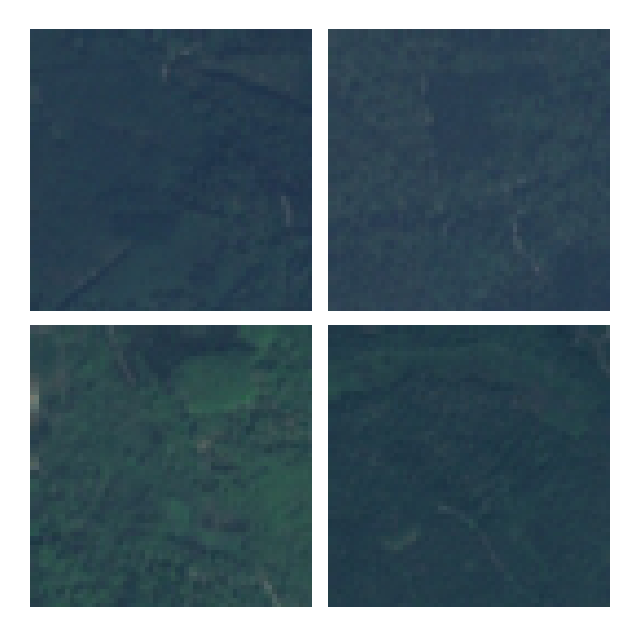}
    \caption*{Forest}
\end{subfigure}
\begin{subfigure}[t]{0.19\linewidth}
    \includegraphics[width=\linewidth]{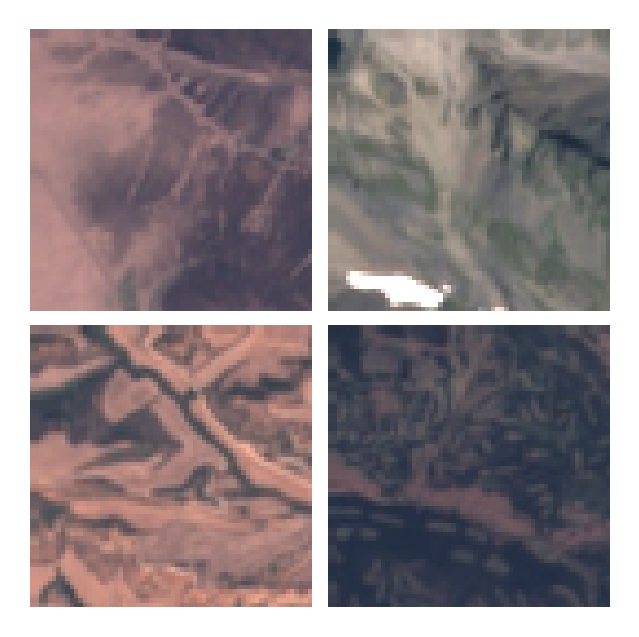}
    \caption*{Herbaceous Veg.}
\end{subfigure}
\begin{subfigure}[t]{0.19\linewidth}
    \includegraphics[width=\linewidth]{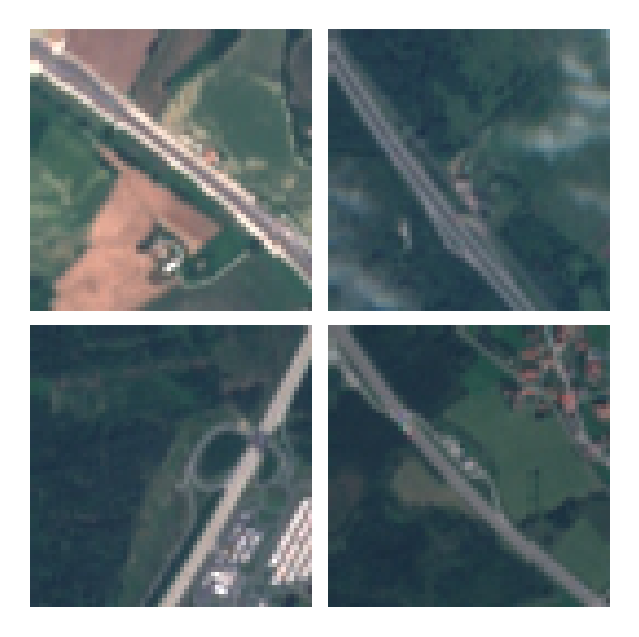}
    \caption*{Highway}
\end{subfigure}
\begin{subfigure}[t]{0.19\linewidth}
    \includegraphics[width=\linewidth]{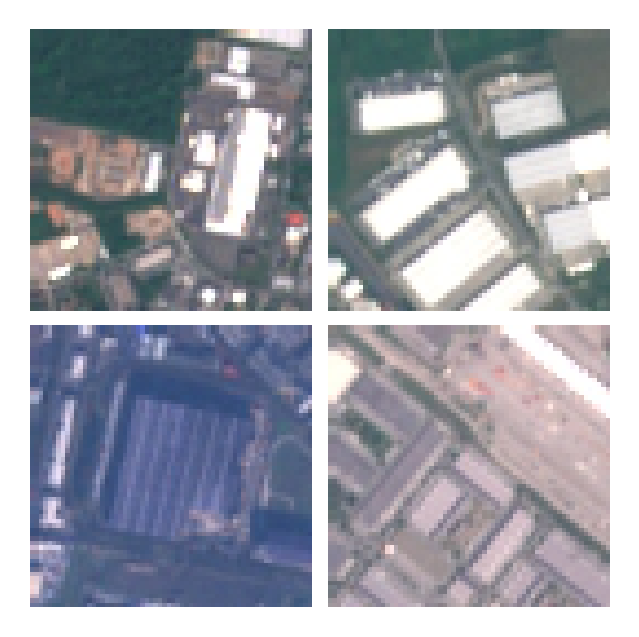}
    \caption*{Industrial}
\end{subfigure}

\par\vspace{0.5em}

\begin{subfigure}[t]{0.19\linewidth}
    \includegraphics[width=\linewidth]{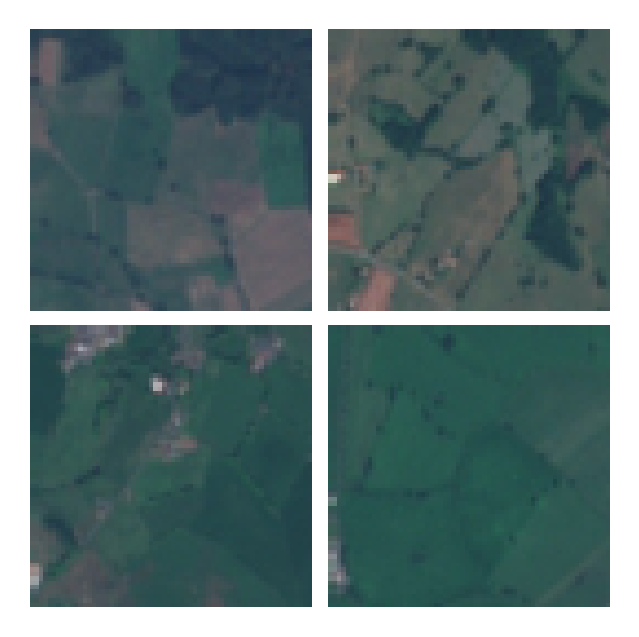}
    \caption*{Pasture}
\end{subfigure}
\begin{subfigure}[t]{0.19\linewidth}
    \includegraphics[width=\linewidth]{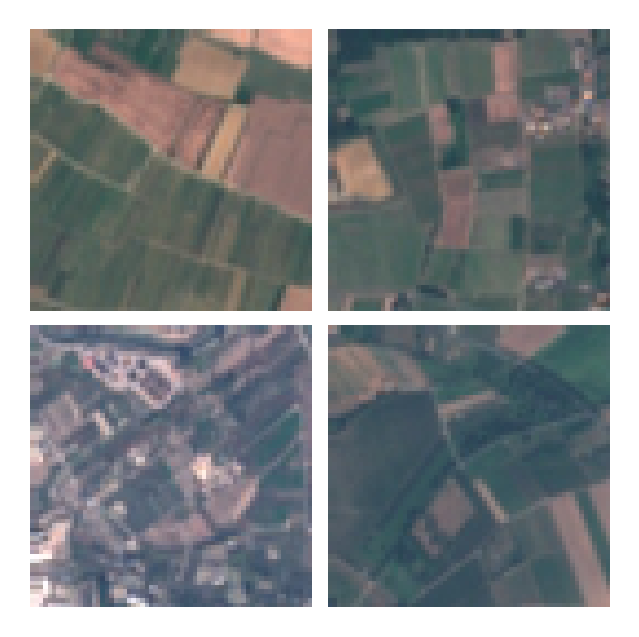}
    \caption*{Permanent Crop}
\end{subfigure}
\begin{subfigure}[t]{0.19\linewidth}
    \includegraphics[width=\linewidth]{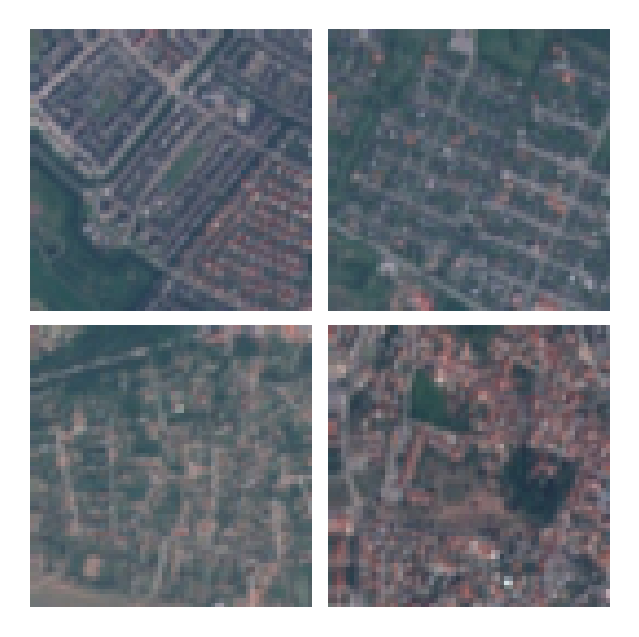}
    \caption*{Residential}
\end{subfigure}
\begin{subfigure}[t]{0.19\linewidth}
    \includegraphics[width=\linewidth]{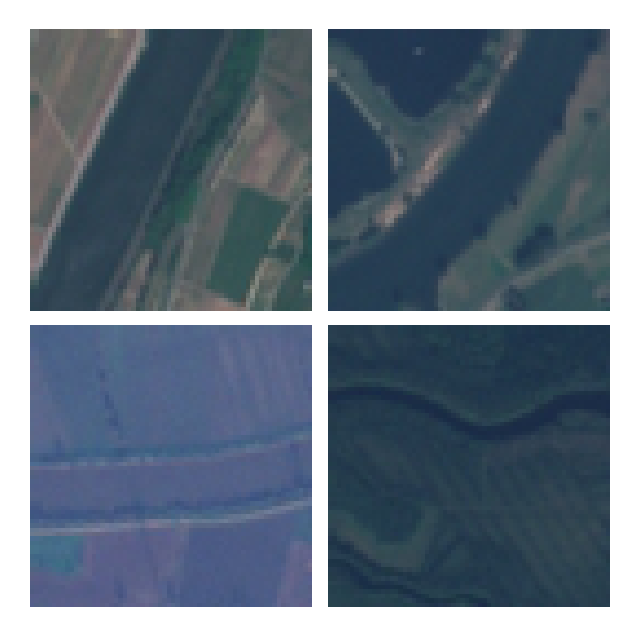}
    \caption*{River}
\end{subfigure}
\begin{subfigure}[t]{0.19\linewidth}
    \includegraphics[width=\linewidth]{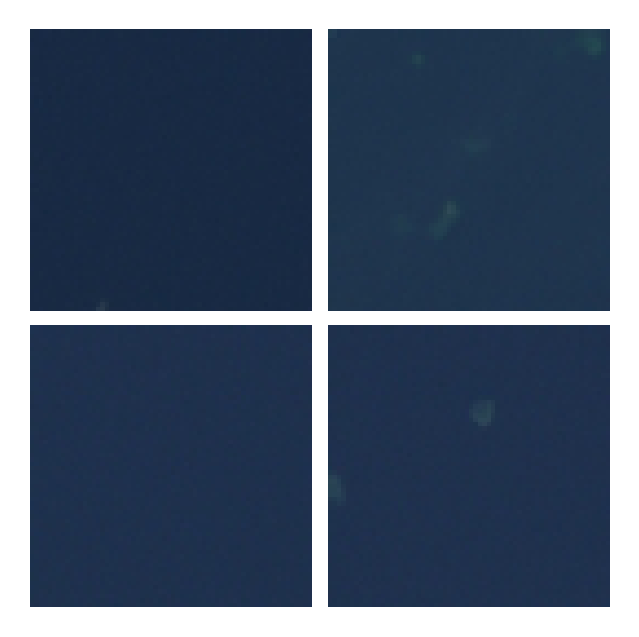}
    \caption*{Sea \& Lake}
\end{subfigure}

\caption{Visualization of the EuroSAT dataset.}
\label{fig:eurosat_samples}
\end{figure}
\subsection{Effect of Successful Link Transmission Parameters}
\label{ELS}

We first investigate how the parameters of the link success probability model affect the overall network connectivity. Specifically, we examine how $\alpha_d$, $\alpha_\theta$, and $w_{ij}$ influence the distribution of the link success probabilities $q_{ij}, \forall i,j$.

To describe the distribution, we use the cumulative distribution function (CDF) of $q_{ij}$. The CDF shows the probability that a randomly selected link has a success probability lower than, or equal to, a given x-axis value. A CDF curve further right indicates that the inter-satellite links have overall higher success probabilities,  implying better overall connectivity.

We evaluate three representative parameter settings:

\begin{itemize}
    \item \textbf{Set A:} $\alpha_d = 0.5$, $\alpha_\theta = 0.7$, $w_{ij} = 0.05$;
    \item \textbf{Set B:} $\alpha_d = 0.7$, $\alpha_\theta = 0.9$, $w_{ij} = 0.05$;
    \item \textbf{Set C:} $\alpha_d = 0.9$, $\alpha_\theta = 0.5$, $w_{ij} = 0.10$.
\end{itemize}

Fig.~\ref{fig:cdf_qij} shows the CDFs of the link success probabilities $q_{ij}$ under the three parameter sets. In this figure, {Set~A} yields the most favorable distribution, with a larger portion of links achieving high $q_{ij}$. {Set~B} shows the steepest rise, indicating more links have degraded success probabilities due to stronger geometric sensitivity. {Set~C} performs in-between, but the larger $w_{ij}$ lowers its overall $q_{ij}$ values further.

\begin{figure}[t]
    \centering
    \includegraphics[width=0.45\textwidth]{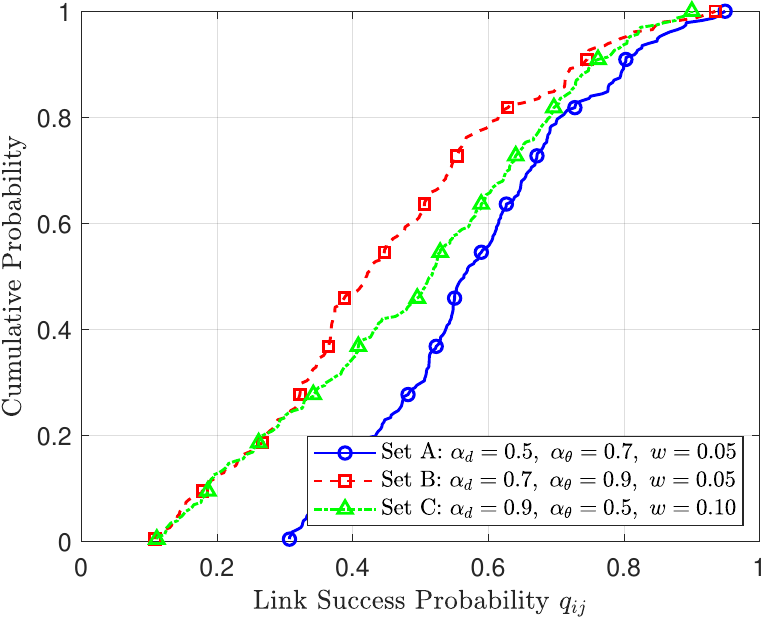}  
    \caption{CDF of link successful transmission probabilities $q_{ij}$ under three parameter configurations.}
    \label{fig:cdf_qij}
\end{figure}
\subsection{Validation of Convergence Analysis}

We validate the convergence analysis presented in Theorem \ref{thm:convergence}, and conduct a set of controlled experiments by directly generating expected mixing matrices $\overline{\mathbf{P}}$ with varying spectral radius $\rho(\overline{\mathbf{P}})$. Specifically, we use the convex optimization tool CVXPY~\cite{diamond2016cvxpy} to synthesize symmetric doubly stochastic matrices with $\rho(\overline{\mathbf{P}}) \in \{0, 0.25, 0.46, 0.74, 0.92\}$. This  allows  isolating the effect of spectral mixing quality from the physical link model.

Figs.~\ref{fig:rho_accuracy_curve} and \ref{fig:rho_accuracy_curve_min} present the average and minimum test accuracy curves over communication rounds under different values of $\rho(\overline{\mathbf{P}})$, respectively. Fig.~\ref{fig:rho_accuracy_curve} shows the global average accuracy across all satellites, while the Fig.~\ref{fig:rho_accuracy_curve_min} reports the worst-case performance, i.e., the lowest test accuracy among all nodes at each round.
From both figures, we observe a  monotonic relationship between $\rho(\overline{\mathbf{P}})$ and the learning performance: A smaller $\rho(\overline{\mathbf{P}})$ (i.e., better network connectivity and faster consensus mixing) leads to significantly faster convergence and better final accuracy. This trend is  pronounced in the minimum accuracy plot, where systems with a large $\rho(\overline{\mathbf{P}})$ suffer from persistent performance gaps across nodes, reflecting poor synchronization and divergent model states. 
These results align with the convergence analysis in Theorem~\ref{thm:convergence}, where the consensus error bound scales with $\rho(\overline{\mathbf{P}})$ through $\Gamma(P^2)$.

\begin{figure}[t]
    \centering
    \includegraphics[width=0.46\textwidth]{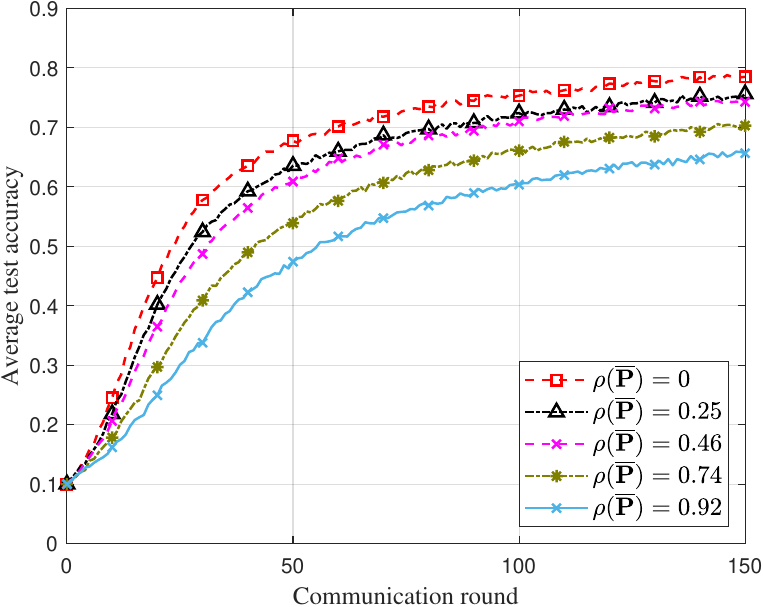}
    \caption{Average test accuracy versus communication round under various $\rho(\overline{\mathbf{P}})$ values generated via CVXPY. }
    \label{fig:rho_accuracy_curve}
\end{figure}
\begin{figure}[t]
    \centering
    \includegraphics[width=0.46\textwidth]{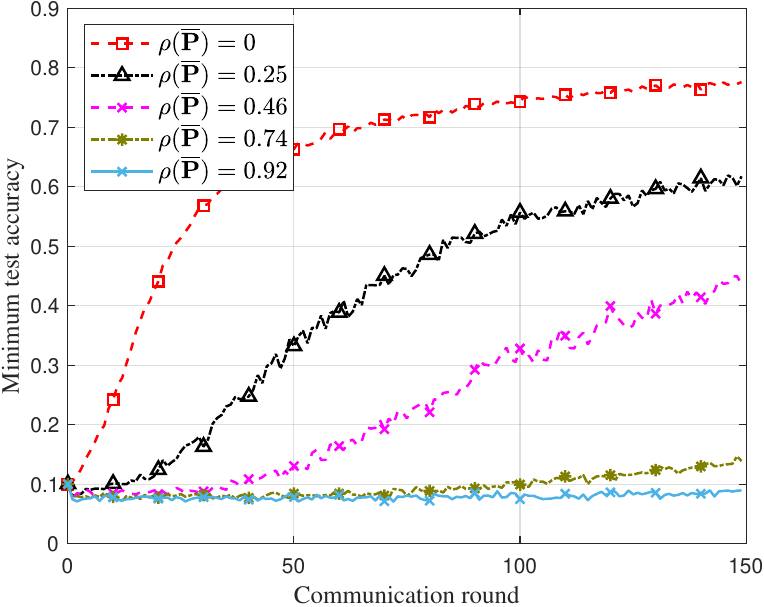}
    \caption{Minimum test accuracy versus communication round under various $\rho(\overline{\mathbf{P}})$ values generated via CVXPY. }
    \label{fig:rho_accuracy_curve_min}
\end{figure}

\subsection{Learning Performance Under Different Settings}

We further examine how different link environments affect the overall learning performance of the proposed subgradient algorithm. Based on the parameter settings defined earlier in Section~\ref{ELS} (i.e., Set~A, Set~B, and Set~C), we simulate the DML process under each setting to evaluate both convergence behavior and final model accuracy.
\begin{figure}[t]
    \centering
    \includegraphics[width=0.46\textwidth]{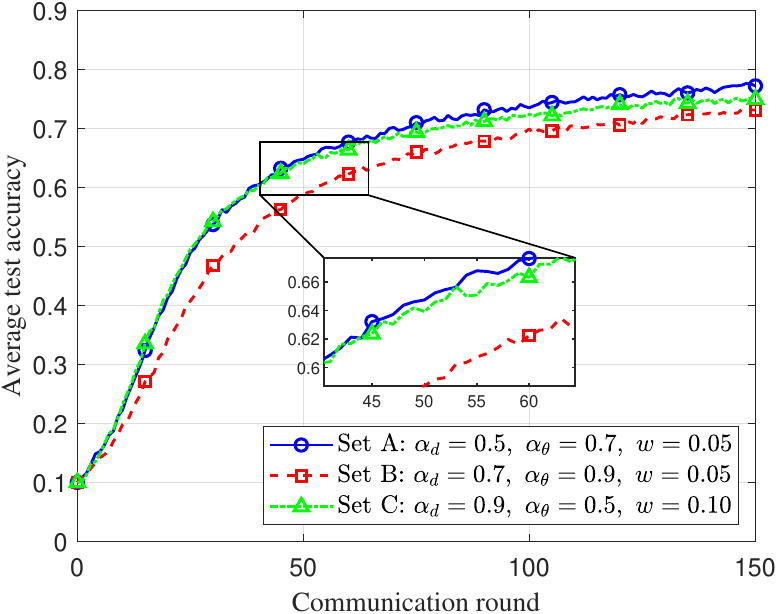}
    \caption{Average test accuracy versus communication round under various  settings. }
    \label{fig:rho_accuracy_curve_diff_setting}
\end{figure}

Fig.~\ref{fig:rho_accuracy_curve_diff_setting} shows the average and minimum test accuracy of the DML over communication rounds. As expected, Set~A yields the best performance, achieving faster convergence and better accuracy across all nodes. This is attributed to its more reliable inter-satellite connectivity, as previously observed in the CDFs of the link success probability. Set~B performs the worst, due to its high sensitivity to both distance and beam deviation, which causes more frequent communication failures and slower model mixing. Set~C exhibits moderate performance but suffers from a degraded reliability floor due to large $w_{ij},\forall i,j$. These observations highlight the critical role of link reliability in decentralized learning: Better connectivity not only accelerates consensus among devices but also improves the final model accuracy. Overall, the proposed method demonstrates robust adaptability across varying link conditions, with performance positively correlated with the quality of the underlying communication topology.

\subsection{Comparison with Benchmarks}

We compare the proposed decentralized subgradient  algorithm against the following state-of-the-art DML strategies adapted from prior work to fit the satellite network setting:

\begin{itemize}
    \item \textbf{Centralized Weight Optimization (CWO) \cite{ye2022decentralized}}: This idealized scheme assumes the existence of a ground-based coordinator that has full access to the inter-satellite link statistics (i.e., the matrix $\mathbf{P}$). It performs the global optimization of aggregation coefficients using the method in~\cite{9716792} and then broadcasts the results  to all satellites.  This method is impractical in LEO systems due to limited Earth access and high latency in practice.

    \item \textbf{Uniform Aggregation without Link Awareness (UWA)~\cite{lin2022distributed}}: Each satellite assigns the same weight to its neighbors, ignoring dynamic probability. This corresponds to using the setting $\mathbf{A} = \frac{1}{N}\mathbf{1}\mathbf{1}^\mathrm{T}$ throughout the training process. Since this benchmark does not adapt to link variability, the resulting $\overline{\mathbf{P}}$ often leads to suboptimal convergence.

    \item \textbf{Fully Reliable Link Approximation (FRLA)}: This method sets ideal communication between all satellites and uses the uniform aggregation weight, which yields $\mathbf{P} = \mathbf{1}\mathbf{1}^\mathrm{T}$ and $\rho(\overline{\mathbf{P}}) = 0$. While this method offers the best-case scenario for the convergence of DML, it fails to reflect the intermittent and directional nature of real-world LEO laser links.

    \item \textbf{Topology-Based Metropolis Weighting (TB-MH) \cite{chib1995understanding}}: Under this scheme, satellite links with success probability below a threshold $q_\delta$ are considered unusable. The resulting  graph defines a static communication topology, over which the aggregation coefficients are computed using the Metropolis-Hastings rule~\cite{robert2004metropolis}. While simple, this heuristic ignores fine-grained probability and requires manual threshold tuning. In our test, we set $q_\delta = 0.8$.
\end{itemize}

\begin{figure}[t]
    \centering
    \includegraphics[width=0.46\textwidth]{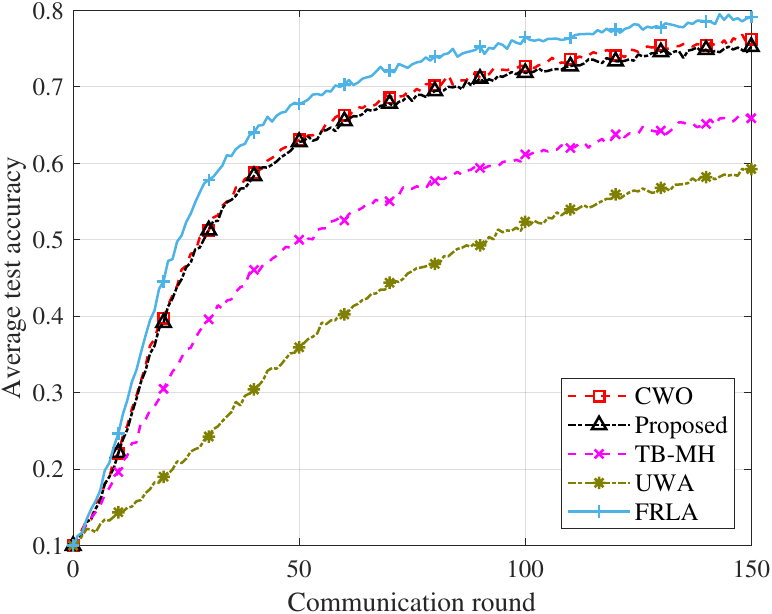}
    \caption{Average test accuracy versus communication round under various  schemes. }
    \label{fig:rho_accuracy_curve_diff_schemes}
\end{figure}

Fig.~\ref{fig:rho_accuracy_curve_diff_schemes} illustrates the convergence behavior of the proposed subgradient algorithm compared with four benchmark strategies in terms of average test accuracy over 150 communication rounds. Among all schemes, {FRLA} achieves the best performance by assuming perfect inter-satellite connectivity and uniform aggregation, thus providing an ideal upper bound.

The proposed subgradient algorithm closely approaches the FRLA performance, especially in the later stages of training, and maintains a small gap throughout. This shows that our decentralized optimization framework—despite operating under stochastic and unreliable links—can achieve near-optimal model mixing efficiency. In addition, it nearly matches the performance of {CWO}, which relies on centralized coordination and global knowledge of the network, highlighting the effectiveness of our fully local subgradient strategy.

{TB-MH} exhibits moderate performance but suffers from degraded convergence due to its rigid link pruning and reliance on static topology. The {UWA} method performs the worst, as it assigns uniform weights regardless of link reliability, leading to inefficient communication and slow learning.
Overall, the proposed subgradient algorithm significantly improves convergence speed and accuracy compared to traditional aggregation strategies, while requiring no global coordination.

\section{Conclusion}

This paper investigated DML under dynamic and unreliable network topologies, where device-to-device communication is subject to probabilistic failures. We modeled the random  link availability through time-varying mixing matrices and formulated  decentralized SGD using a compact matrix representation. To characterize the asymptotic convergence, we derived theoretical bounds that explicitly depend on the second-order statistics of the mixing process and proposed a tractable surrogate optimization objective based on the spectral radius of the expected mixing matrix. We developed a fully decentralized algorithm to minimize this objective, which includes local subgradient updates, distributed eigenvector estimation via Chebyshev acceleration, and a symmetric normalization mechanism to ensure feasibility. Simulation results on the EuroSAT dataset confirmed that the proposed decentralized optimization method  accelerates training convergence and improves model accuracy. These findings validate the applicability of the method in large-scale  networks and suggest its potential for broader DML deployment.

\ifCLASSOPTIONcaptionsoff
  \newpage
\fi


\bibliographystyle{IEEEtran}
\bibliography{references}

\end{document}